\documentclass[10pt,conference,letter]{IEEEtran}

\usepackage{times}

\usepackage{amssymb}
\usepackage{amsthm}
\usepackage{amsxtra}
\usepackage{amsmath}

\usepackage{verbatim}
\usepackage{epsfig}
\usepackage{setspace}
\usepackage{caption}
\usepackage{subcaption}
\usepackage{breqn}

\allowdisplaybreaks

\newtheorem{theorem}{Theorem}
\newtheorem{definition}{Definition}
\newtheorem{lemma}[theorem]{Lemma}
\newtheorem{proposition}[theorem]{Proposition}
\newtheorem{corollary}{Corollary}

\newtheorem{conjecture}{Conjecture}

\def\Pr{\text{Pr}}

\def\bF{\mathbf{F}}

\def\bR{\mathbf{R}}

\def\cE{\mathcal{E}}
\def\cF{\mathcal{F}}
\def\cG{\mathcal{G}}

\def\cP{\mathcal{P}}
\def\cR{\mathcal{R}}
\def\cS{\mathcal{S}}

\def\cV{\mathcal{V}}
\def\cX{\mathcal{X}}

\def\N{\mathbb{N}}

\def\pl{\parallel}

\def\RTmin{R_{T}^{(\min)}}
\def\IXm{\mathbf{I}(X_{[m]})}

\def\beq{\begin{equation}}
\def\eeq{\end{equation}}

\title{Achieving SK Capacity in the Source Model: \\ When Must All Terminals Talk?}
\author{
\IEEEauthorblockN{Manuj Mukherjee$^\dag$} \and \IEEEauthorblockN{Navin Kashyap$^\dag$} \and \IEEEauthorblockN{Yogesh Sankarasubramaniam$^\ddag$}
}

\begin{document}

\maketitle

\renewcommand{\thefootnote}{}
\footnotetext{
\noindent $^\dag$M.\ Mukherjee and N.\ Kashyap are with the 
Department of Electrical Communication Engineering, 
Indian Institute of Science, Bangalore. Email: \{manuj,nkashyap\}@ece.iisc.ernet.in.

\smallskip

$^\ddag$Email: yogesh@gatech.edu
}

\begin{abstract}
 In this paper, we address the problem of characterizing the instances of the multiterminal source model of Csisz{\'a}r and Narayan in which communication from \emph{all} terminals is needed for establishing a secret key of maximum rate. We give an information-theoretic sufficient condition for identifying such instances. We believe that our sufficient condition is in fact an exact characterization, but we are only able to prove this in the case of the three-terminal source model. We also give a relatively simple criterion for determining whether or not our condition holds for a given multiterminal source model. 
\end{abstract}

\renewcommand{\thefootnote}{\arabic{footnote}}

\section{Introduction} \label{sec:intro}
We are concerned with the multiterminal source model of Csisz{\'a}r and Narayan \cite{CN04}, which can be briefly described as follows. There are a certain number, $m \ge 2$, of terminals, each of which observes a distinct component of a source of correlated randomness. The terminals must agree on a shared SK by communicating over a noiseless public channel. This key must be protected from a passive eavesdropper having access to the public communication. The SK capacity, which is the supremum of the rates of SKs that can be generated, has been characterized in various ways \cite{CZ10,CN04,NN10}. What is less well-understood is the nature of public communication that is needed to achieve SK capacity in this model. In a companion paper \cite{MK14}, we gave a lower bound on the minimum rate of communication required to generate a maximal-rate (i.e., capacity-achieving) SK, building upon the prior work of Tyagi \cite{Tyagi13} on the two-terminal model. In this paper, we address a related question: when must all $m$ terminals necessarily have to communicate in order to generate a maximal-rate SK? 

It is well known that, in order to generate a maximal-rate SK in the two-terminal model ($m=2$), it is sufficient for only one terminal to communicate \cite{AC93, Maurer93,CN04}. All this terminal has to do is convey its local observations to the other terminal at the least possible rate of communication required to do so. Thus, when $m=2$, it is \emph{never necessary} for both terminals to communicate to generate a capacity-achieving SK. Even when $m > 2$, there are examples wherein not all terminals need to communicate --- see remark following Theorem~1 in \cite{CN04}. However, as we will show in this paper, there are plenty of other examples where all terminals \emph{must} communicate in order to achieve SK capacity. We coin the term ``omnivocality'' to describe the state when all terminals communicate. The problem of interest to us is the following: \emph{characterize the instances of the multiterminal source model in which omnivocality is necessary for maximal-rate SK generation}. In this paper, we report partial progress made towards such a characterization. 

The paper is organized as follows. After establishing the required notation and background in Section~\ref{sec:prelim}, we give, in Section~\ref{sec:omni}, a sufficient condition under which omnivocality is necessary for achieving SK capacity in a source model with $m \ge 3$ terminals. This condition is satisfied, for example, in the case of the complete graph pairwise independent network (PIN) model of Nitinawarat and Narayan \cite{NN10}. We conjecture that our sufficient condition is also necessary, but at present, we can only prove this in the $m=3$ case. Finally, in Section~\ref{sec:singleton}, we give a useful criterion for checking whether or not our condition holds for a given source model.

\section{Preliminaries} \label{sec:prelim}

 Throughout, we use $\N$ to denote the set of positive integers. In the multiterminal source model \cite{CN04}, a set of $m\ge 2$ terminals, denoted by $[m]\triangleq\{1,2, \ldots, m\}$, has access to a source $(X_1^n,X_2^n,\ldots,X_m^n)$, $n \in \N$, where $X_i^n$ denotes $n$ i.i.d.\ copies of a random variable (rv) $X_i$ taking values in a finite set $\cX_i$. The rvs $X_1,X_2,\ldots,X_m$ are in general correlated, and for each $i \in [m]$, the $i$th terminal observes only the component $X_i^n$. For any subset $A\subseteq [m]$, we will use $X_A$ to denote the collection of rvs $(X_i:i \in A)$, and $p_{X_A}$ to denote their joint probability mass function. 

The terminals communicate through a noiseless public channel, any communication sent through which is accessible to all terminals and to potential eavesdroppers as well. The terminals communicate in a round-robin fashion, following the cyclic order $(1,2,\ldots,m)$. Any transmission sent by the $i$th terminal is a deterministic function of $X_i^n$ and all the previous communication. Formally, a \emph{valid communication} is a finitely-supported random vector $\mathbf{F} = (F_1,F_2,\ldots,F_r)$, $r \in \N$, with $F_j$ denoting a communication sent by the terminal $i \in [m]$ with $i \equiv j \!\! \pmod m$, and $H(F_j \mid F_1,\ldots,F_{j-1},X_i^n) = 0$. The \emph{rate} of the communication is taken to be $\frac{1}{n} \log_2|\cF|$, where $\cF$ is the finite set on which $\mathbf{F}$ is supported. Terminal $i \in [m]$ is said to be \emph{silent} if $F_j = 0$ (with probability $1$) for all $j \equiv i \!\! \pmod{m}$. An \emph{omnivocal} communication is one in which no terminal is silent. 

Given an $\epsilon > 0$, we say that an rv $U$ is \emph{$\epsilon$-recoverable} from an rv $V$ if there exists a function $g$ of $V$ such that $\Pr[U = g(V)] \ge 1 - \epsilon$. 
\begin{definition}
\label{def:SK}
For any $\epsilon > 0$, an \emph{$\epsilon$-SK} for $[m]$ is an rv $K = K^{(n)}(X_{[m]}^n)$, for some $n \in \N$, such that there exists a valid communication $\mathbf{F}$ with the following properties: \\
\qquad \emph{(i)} $I(K;\mathbf{F}) \le \epsilon$; and \\
\qquad \emph{(ii)} $K$ is $\epsilon$-recoverable from $(X_i^n,\mathbf{F})$ for each $i \in [m]$. \\
The \emph{rate} of this $\epsilon$-SK is given by $\frac{1}{n}H(K)$. 
\end{definition}

A real number $R\geq 0$ is an \emph{achievable SK rate} if for any $\epsilon > 0$, there exists an $\epsilon$-SK of rate greater than $R - \epsilon$. The \emph{SK capacity $C([m])$} is defined as the supremum of all achievable SK rates. The SK capacity can be expressed as \cite[Theorem~1.1]{CZ10} (see also \cite[Eq.~(26)]{CN04})
\begin{equation}
C([m]) = \IXm \triangleq \min_{\cP} \frac{1}{|\cP|-1} D\left(p_{X_{[m]}} \pl \prod_{A \in \cP} p_{X_A} \right)
\label{eq:I}
\end{equation}
the minimum being taken over all partitions $\cP$ of $[m]$, of size ${|\cP|} \ge 2$. The quantity $D(\cdot \| \cdot)$ denotes relative entropy, and for a partition $\cP = \{A_1,\ldots,A_k\}$, the notation $\prod_{A\in\cP}p_{X_A}$ represents the product $p_{X_{A_1}} \times \cdots \times p_{X_{A_k}}$. Note that when $m=2$, the quantity $\IXm$ defined in \eqref{eq:I} simply reduces to the mutual information $I(X_1;X_2)$. Thus, $\IXm$ should be viewed as a multiparty extension of mutual information. 

Before proceeding further, a couple of clarifications concerning Definition~\ref{def:SK} are needed. We have adopted the notion of \emph{strong secrecy} (property~(i) in the definition), as opposed to \emph{weak secrecy}, which only requires $\frac{1}{n}I(K;\bF) \le \epsilon$. All the results proved in this paper would hold just as well under either type of secrecy. In particular, our main result shows that omnivocal communication is necessary for achieving SK capacity if a certain condition on the singleton partition $\cS$ is satisfied. Our proof of this result relies only on the expression for SK capacity given in \eqref{eq:I}, which remains the same under both forms of secrecy \cite{CN04}, and on a theorem of Gohari and Anantharam \cite{GA10}, which is stated and proved under the weak secrecy notion. Thus, our proof in fact shows that omnivocal communication is necessary even under a weak secrecy requirement on SKs.

A second clarification concerning Definition~\ref{def:SK} is that, usually, the definition of an $\epsilon$-SK  includes an additional requirement that $K$ be almost uniformly distributed over its alphabet $\mathcal{K}$, i.e., $H(K) \ge \log |\mathcal{K}| - \epsilon$ \cite{CN04}. However, this can always be dropped without affecting SK capacity --- see e.g., the discussion on p.~3976 in \cite{GA10}.

As mentioned above, we make use of a result of Gohari and Anantharam \cite[Theorem~6]{GA10} in some of our proofs. To state this result, we explicitly define a \emph{weak $\epsilon$-SK} for $[m]$ to be an rv $K$ as in Definition~\ref{def:SK}, except that the strong secrecy condition~(i) is replaced by the weak secrecy condition, $\frac{1}{n}I(K;\bF) \le \epsilon$. Then, $R \ge 0$ is an \emph{achievable weak-SK rate} if for any $\epsilon > 0$, there exists a weak $\epsilon$-SK of rate greater than $R - \epsilon$. It is known that the supremum of achievable weak-SK rates is the same as the SK capacity given by \eqref{eq:I}. The Gohari-Anantharam result concerns achievable weak-SK rates under the additional assumption that some fixed subset of terminals remains silent throughout. Let $T \subseteq [m]$ be such that terminals in $T$ are allowed to communicate, while terminals in $[m] \setminus T$ must remain silent. Thus, we are restricted to valid communications $\bF$ in which the terminals in $[m] \setminus T$ are silent, but which allow all $m$ terminals to agree upon a weak SK. In other words, we only consider weak $\epsilon$-SKs for $[m]$ that are obtainable through valid communications $\bF$ in which all terminals in $[m] \setminus T$ are silent. The supremum of rates achievable by such SKs will be denoted by $C([m] \| T)$. 

\begin{theorem}[\cite{GA10}, Theorem~6]
$C([m] \| T) = H(X_{T})-R_{T}^{(\min)}$, where $\RTmin=\displaystyle\min_{\bR\in\cR_{T}} \sum_{i \in T} R_i$, the rate region $\cR_T$ being the set of all points $\bR = (R_i, i \in T)$ such that 
$$
\sum_{i\in A \cap T} R_i\geq H(X_{A \cap T}|X_{A^c}) \quad \forall \, A \subsetneq [m], \ A \cap T\neq\emptyset.
$$
\label{th:silent}
\end{theorem}

Note that if $C([m]) > C([m] \| T)$ for all $T \subset [m]$ of size ${|T|}=m-1$, then omnivocality is necessary for achieving SK capacity. Thus, our approach for showing that omnivocal communication is needed in certain cases is to use Theorem~\ref{th:silent} to prove that $C([m]) > C([m] \| T)$ for all $(m-1)$-subsets $T \subset [m]$. For this, we will need a lower bound on $\RTmin$ when ${|T|} = m-1$. To prove this bound, we use a simpler characterization (than that given in Theorem~\ref{th:silent}) of the rate region $\cR_T$ when ${|T|}=m-1$.

\begin{lemma} Let $T = [m] \setminus \{u\}$ for some $u \in [m]$. 
The rate region $\cR_{T}$ is the set of all points $(R_i, \, i \in T)$ such that
\beq
\sum_{i \in B} R_i \ge H(X_{B} | X_{T \setminus B}) \quad \forall \, B \subsetneq T, B \ne \emptyset,
\label{eq:lem}
\eeq
and 
$\displaystyle \sum_{i\in T} R_i \ge H(X_T| X_u)$.
\label{lem:Rm-1}
\end{lemma}
\begin{IEEEproof}
Observe that $\cR_T$ is defined by constraints on sums of the form $\sum_{i \in B} R_i$ for non-empty subsets $B \subseteq T$. When $B = T$, the constraint is simply $\sum_{i \in T} R_i \ge H(X_T| X_u)$.

Now, consider any non-empty $B \subsetneq T$. From Theorem~\ref{th:silent}, we see that constraints on $\sum_{i \in B} R_i$ arise as constraints on $\sum_{i \in A \cap T} R_i$ in two ways: when $A = B$ and when $A = B \cup \{u\}$.  Thus, we have two constraints on $\sum_{i \in B} R_i$: 
$$
\sum_{i \in B} R_i \ge H(X_{B} | X_{[m] \setminus B}),
$$
obtained when $A = B$, and 
$$
\sum_{i \in B} R_i \ge H(X_{B} | X_{T \setminus B}),
$$
obtained when $A = B \cup \{u\}$. The latter constraint is clearly stronger, so we can safely discard the former.
\end{IEEEproof}

We can now prove the desired lower bound on $\RTmin$.

\begin{lemma}
Let $m \ge 3$ be given. For $T \subset [m]$ with ${|T|} = m-1$, we have 
$$
\RTmin \ge \frac{1}{m-2} \sum_{j \in T} H(X_{T \setminus \{j\}} | X_j).
$$
\label{lem:RTmin}
\end{lemma}
\begin{IEEEproof}
Consider any $T \subset [m]$  with ${|T|} = m-1$. For each $j \in T$, let $B_j = T \setminus \{j\}$. Now, let $(R_i, i \in T)$ be any point in $\cR_T$. Applying \eqref{eq:lem} with $B = B_j$, we get
$$
\sum_{i \in B_j} R_i \ge H(X_{T\setminus \{j\}} | X_j),
$$
for each $j \in T$. Summing over all $j \in T$, we obtain
\beq
\sum_{j \in T} \sum_{i \in B_j} R_i \ge \sum_{j \in T} H(X_{T\setminus \{j\}} | X_j).
\label{RTmin_eqa}
\eeq
Exchanging the order of summation in the double sum on the left-hand side (LHS) above, we have $\sum_{j \in T} \sum_{i \in B_j} R_i = 
\sum_{i \in T} \sum_{j \in B_i} R_i  = \sum_{i \in T} (m-2)R_i  =  (m-2) \sum_{i \in T} R_i$.
Putting this back into \eqref{RTmin_eqa}, we get 
$$
\sum_{i \in T} R_i \ge \frac{1}{m-2} \sum_{j \in T} H(X_{T\setminus \{j\}} | X_j).
$$
Since this holds for any $(R_i, i \in T) \in \cR_T$, the lemma follows.
\end{IEEEproof}

\section{Omnivocal Communication} \label{sec:omni}

As pointed out in the Introduction, in the source model with two terminals, omnivocality is never necessary for generating a maximal-rate SK. However, the situation is different when there are three or more terminals. In this section, we give a sufficient condition for omnivocality being needed for achieving SK capacity when there are $m \ge 3$ terminals, and give an example where the sufficient condition is met. The sufficient condition also turns out to be necessary when there are exactly three terminals. 

To state our results, we need a few definitions. The partition $\bigl\{\{1\},\{2\},\ldots,\{m\}\bigr\}$ consisting of $m$ singleton cells will play a special role in our results; we call this the \emph{singleton partition} and denote it by $\cS$. For any partition $\cP$ of $[m]$ with ${|\cP|} \ge 2$, define
\begin{equation}
\Delta(\cP) \triangleq \frac{1}{|\cP|-1} \left[\sum_{A \in \cP} H(X_A) - H(X_{[m]}) \right].
\label{def:DP}
\end{equation}
Equivalently,
\begin{equation*}
\Delta(\cP) = \frac{1}{|\cP|-1} D\left(p_{X_{[m]}} \pl \prod_{A \in \cP} p_{X_A} \right), 
\end{equation*}
the notation being as in \eqref{eq:I}. Thus, $C([m]) = \mathbf{I}(X_{[m]}) = \min_{\cP}\Delta(\cP)$. In all that follows, we say that the singleton partition $\cS$ is a \emph{minimizer for $\IXm$} if $\Delta(\cS) = \IXm$, and that $\cS$ is the \emph{unique minimizer for $\mathbf{I}(X_{[m]})$}  if the minimum in \eqref{eq:I} is uniquely achieved by $\cS$, i.e., $\Delta(\cS) < \Delta(\cP)$ for all partitions $\cP$ of $[m]$, $\cP \ne \cS$, with ${|\cP|} \ge 2$. 

We can now state the main result of this section.

\begin{theorem}
For $m \ge 3$ terminals, if $\cS$ is the unique minimizer for $\mathbf{I}(X_{[m]})$,
then omnivocal communication is necessary for achieving the SK capacity $C([m])$. 
\label{th:mge3}
\end{theorem}

Before proving the theorem, we give an example where the condition of the theorem is met.

The pairwise independent network (PIN) model of Nitinawarat and Narayan \cite{NN10} is defined on an underlying graph $\cG = (\cV,\cE)$ with $\cV = [m]$. For $n \in \N$, let $\cG^{(n)}$ be the multigraph $(\cV,\cE^{(n)})$, where $\cE^{(n)}$ is the multiset of edges formed by taking $n$ copies of each edge of $\cG$. Associated with each edge $e \in \cE^{(n)}$ is a Bernoulli$(1/2)$ rv $\xi_e$; the rvs $\xi_e$ associated with distinct edges in $\cE^{(n)}$ are independent. With this, the rvs $X_i^n$, $i\in [m]$, are defined as $X_i^n=(\xi_e: e\in\cE^{(n)} \text{ and } e \text{ is incident on } i)$. When $\cG = K_m$, the complete graph on $m$ vertices, we have the \emph{complete graph PIN model}. 
 
We show in the next section (Corollary~\ref{cor:pin}) that for the complete graph PIN model, the singleton partition $\cS$ is the unique minimizer for $\IXm$. The result below then immediately follows from Theorem~\ref{th:mge3}.

\setcounter{corollary}{0}
\begin{corollary}
In the PIN model defined on the complete graph $K_m$, $m \ge 3$, omnivocal communication is necessary for achieving $C(X_{[m]})$.
\label{cor2}
\end{corollary}

In conjunction with Theorem~6 in \cite{MK14}, we now have the following picture for a capacity-achieving communication in the complete graph PIN model: the communication must be omnivocal, and if it is constrained to be a linear function of the observations $X_{[m]}^n$, then it must have rate at least $m(m-2)/2$. It should be noted that the capacity-achieving communication in the proof of \cite[Theorem~1]{NN10} is an omnivocal, linear communication of rate $m(m-2)/2$.

\medskip
For the proof of Theorem~\ref{th:mge3}, we need some convenient notation. For $T \subset [m]$, ${|T|} = m-1$, define $\Delta_T(\cS) \triangleq \frac{1}{m-2}[\sum_{i \in T} H(X_i) - H(X_T)]$.

\begin{lemma}
For $m \ge 3$ terminals, if the singleton partition $\cS$ is the unique minimizer for $\IXm$, then $\Delta_T(\cS)  <  \Delta(\cS)$ for all $T \subset [m]$ with ${|T|} = m-1$.
\label{lem:Delta}
\end{lemma}
\begin{IEEEproof}
For any $u \in [m]$, consider $T = [m] \setminus \{u\}$. Using $\Delta(\cS) = \frac{1}{m-1} [\sum_{i=1}^m H(X_i) - H(X_{[m]})]$ and the definition of $\Delta_T(\cS)$ above, it is easy to verify the identity
$$
{\textstyle
\frac{m-1}{m-2} \Delta(\cS) = \Delta_T(\cS) + \frac{1}{m-2} I(X_u;X_T).
}
$$
Re-arranging the above, we obtain
\begin{align}
\Delta_T(\cS) - \Delta(\cS) & = {\textstyle \frac{1}{m-2}}[\Delta(\cS) - I(X_u;X_T)] \notag \\
& = {\textstyle \frac{1}{m-2}} [\Delta(\cS) - \Delta(\cP)], \label{Ddiff}
\end{align}
where $\cP$ is the 2-cell partition $\bigl\{\{u\},T \bigr\}$ of $[m]$. By assumption, the expression in \eqref{Ddiff} is strictly negative.
\end{IEEEproof}

\medskip

With this, we are ready to prove Theorem~\ref{th:mge3}.

\begin{IEEEproof}[Proof of Theorem~\ref{th:mge3}]
We will show that $C([m]) > C([m] \| T)$ for any $T \subset [m]$ with ${|T|}=m-1$. First, note that since $\cS$ is, by assumption, a minimizer for $\IXm$, we have $C([m]) = \IXm = \Delta(\cS)$. Next, by Theorem~\ref{th:silent} and Lemma~\ref{lem:RTmin}, we have
\begin{align*}
C([m] \| T)  & \le H(X_T) - {\textstyle \frac{1}{m-2}} \sum_{i\in T}H(X_{T\setminus\{i\}}|X_i), \notag \\
& = {\textstyle \frac{1}{m-2}} \biggl[(m-2) H(X_T) - \sum_{i\in T} [H(X_T) - H(X_i)] \biggr] \notag \\
& = \Delta_T(\cS).
\end{align*}
Therefore, $C([m] \| T) \le \Delta_T(\cS) < \Delta(\cS) = C([m])$, the second inequality coming from Lemma~\ref{lem:Delta}. 
\end{IEEEproof}

\medskip

For the three-terminal source model, it turns out that the unique minimizer condition in Theorem~\ref{th:mge3} is also necessary for the conclusion of the theorem to hold. Note that when $m=3$, \eqref{eq:I} reduces to $C(X_{[3]}) = \min\bigl\{I(X_{\{1,2\}};X_3), I(X_{\{1,3\}};X_2),  I(X_{\{2,3\}};X_1),  \Delta(\cS)\bigr\}$; so the unique minimizer condition is equivalent to 
$$\Delta(\cS) < \min\{I(X_{\{1,2\}};X_3), I(X_{\{1,3\}};X_2), I(X_{\{2,3\}};X_1)\}.$$

\begin{theorem}
In the three-terminal source model, omnivocal communication is necessary for achieving SK capacity iff the singleton partition $\cS$ is the unique minimizer for $\IXm$.
\label{th:meq3}
\end{theorem}
\begin{IEEEproof}
The ``if'' part is by Theorem~\ref{th:mge3}. For the ``only if'' part, suppose that $\Delta(\cS) \ge \min\{I(X_{\{1,2\}};X_3), I(X_{\{1,3\}};X_2), I(X_{\{2,3\}};X_1)\}$. Then, $\Delta(\cS)$ is either (a)~greater than or equal to at least two of the three terms in the minimum, or (b)~greater than or equal to exactly one term. Up to symmetry, it suffices to distinguish between two cases:

 Case I: $\Delta(\cS) \ge \max\{I(X_{\{1,2\}};X_3), I(X_{\{1,3\}};X_2)\}$

 Case II: $\min\{I(X_{\{1,3\}};X_2), I(X_{\{2,3\}};X_1)\} > \Delta(\cS) \ge I(X_{\{1,2\}};X_3)$ 

\noindent In each case, we demonstrate a capacity-achieving communication in which at least one terminal remains silent.

\medskip

We deal with Case~I first. Observe that $\Delta(\cS) = \frac12\left[\sum_{i=1}^3 H(X_i) - H(X_{[3]})\right]$ can also be written as $\frac12[I(X_1;X_2) + I(X_{\{1,2\}};X_3)]$. Thus, the assumption $\Delta(\cS) \ge I(X_{\{1,2\}};X_3)$, upon some re-organization, yields $I(X_1;X_2) \ge I(X_{\{1,2\}};X_3)$, i.e.,
\begin{equation}
I(X_1;X_2) \ge I(X_1;X_3) + I(X_2;X_3 | X_1).
\label{case1:eqa}
\end{equation}
 Similarly, using the identity $\Delta(\cS)=\frac12[I(X_1;X_3) + I(X_{\{1,3\}};X_2)]$ in the assumption $\Delta(\cS) \ge I(X_{\{1,3\}};X_2)$, we obtain $I(X_1;X_3) \ge I(X_{\{1,3\}};X_2)$, i.e.,
\begin{equation}
I(X_1;X_3) \ge I(X_1;X_2) + I(X_1;X_3 | X_2).
\label{case1:eqb}
\end{equation}
The equalities in \eqref{case1:eqa} and \eqref{case1:eqb} can simultaneously hold iff 
\begin{equation}
\begin{gathered}
I(X_1;X_2) = I(X_1;X_3) \ \ \text{ and } \\
I(X_1;X_3|X_2) = I(X_2;X_3|X_1) = 0.
\end{gathered}
\label{case1:eqc}
\end{equation}
From \eqref{case1:eqc}, it is not hard to deduce that the quantities $I(X_{\{1,2\}};X_3)$, $I(X_{\{1,3\}};X_2)$, $I(X_{\{2,3\}};X_1)$ and $\Delta(\cS)$ are all equal to $I(X_1;X_2)$. In particular, $C(X_{[3]}) = I(X_1;X_2)$.

From the first equality in \eqref{case1:eqc}, we also have $H(X_1|X_2) = H(X_1|X_3)$. Now, it can be shown by a standard random binning argument that there exists a communication from terminal $1$ of rate $H(X_1|X_2) = H(X_1|X_3)$ such that $X_1^n$ is $\epsilon$-recoverable at both terminals $2$ and $3$. It then follows from the ``balanced coloring lemma'' \cite[Lemma~B.3]{CN04} that an SK rate of $H(X_1) - H(X_1|X_2) = I(X_1;X_2)$ is achievable. Thus, the SK capacity, $C([3]) = I(X_1;X_2)$, is achievable by a communication in which terminals $2$ and $3$ are both silent.

\medskip

Now, consider Case~II, in which we obviously have $C([3]) = I(X_{\{1,2\}};X_3)$. The idea here is to show that a valid communication of rate $H(X_{\{1,2\}}|X_3)$ exists in which terminal $3$ is silent, and which allows $\epsilon$-recoverability of $(X_1^n,X_2^n)$ at all three terminals. Given this, an application of \cite[Lemma~B.3]{CN04} shows that an SK rate of $H(X_{\{1,2\}}) - H(X_{\{1,2\}}|X_3) = I(X_{\{1,2\}};X_3)$ is achievable. Thus, there is a $C([3])$-achieving communication in which terminal $3$ is silent. 

To show that the desired communication exists, we argue as follows. For $i=1,2$, let $R_i$ be the rate at which terminal $i$ communicates. A standard random binning argument shows that an achievable $(R_1,R_2)$ region, with terminal $3$ silent, for a communication intended to allow $\epsilon$-recoverability of $(X_1^n,X_2^n)$ at all terminals is given by 
\begin{equation}
\begin{gathered}
R_1 \ge H(X_1 | X_2), \ \ R_2 \ge H(X_2|X_1), \\
R_1+R_2 \ge H(X_{\{1,2\}}|X_3).
\end{gathered}
\label{case2:eqa}
\end{equation}

Now, using the assumption in Case~II that $\Delta(\cS) \ge I(X_{\{1,2\}};X_3)$, we will prove that the inequality 
\beq
H(X_1 | X_2) + H(X_2|X_1) \le H(X_{\{1,2\}}|X_3)
\label{case2:eqb}
\eeq
holds. It would then follow from \eqref{case2:eqa} that there exist achievable rate pairs $(R_1,R_2)$ with $R_1 + R_2 = H(X_{\{1,2\}}|X_3)$, thus completing the proof for Case~II. 

So, let us prove \eqref{case2:eqb}. We have $\Delta(\cS) =  \frac12[H(X_1) + H(X_2) + H(X_3) - H(X_{[3]})]$ and $I(X_{\{1,2\}};X_3) = H(X_{\{1,2\}})+H(X_3)-H(X_{[3]})$. Using these expressions in the inequality $\Delta(\cS) \ge I(X_{\{1,2\}};X_3)$, and re-arranging terms, we obtain
$$
\frac12[H(X_1)+H(X_2) - 2H(X_{\{1,2\}})] \ge \frac12[H(X_3) - H(X_{[3]})],
$$
which is equivalent to \eqref{case2:eqb}. This completes the proof of the theorem.
\end{IEEEproof}

We in fact conjecture that the result of Theorem~\ref{th:meq3} should extend to more than three terminals as well.

\begin{conjecture}
In the multiterminal source model with $m \ge 3$ terminals, omnivocal communication is necessary for achieving SK capacity iff the singleton partition is the unique minimizer for $\mathbf{I}(X_{[m]})$.
\label{conj}
\end{conjecture}

At this point, we do not have a systematic approach for proving the ``only if'' part of the conjecture for $m \ge 4$. 

\section{Singleton Partitions} \label{sec:singleton} 

The condition that the singleton partition be a unique minimizer for $\mathbf{I}(X_{[m]})$ plays a key role in the results of Section~\ref{sec:omni}. Thus, it would be very useful to have a way of checking whether this condition holds for a given source $X_{[m]}$, $m \ge 3$. The brute force method of comparing $\Delta(\cS)$ with $\Delta(\cP)$ for all partitions $\cP$ with at least two parts requires an enormous amount of computation. Indeed, the number of partitions of an $m$-element set is the $m$th Bell number, $B_m$, an asymptotic estimate for which is $(\log w)^{1/2} w^{m-w}e^w$, where $w = \frac{m}{\log m} \, [1+o(1)]$ is the solution to the equation $m = w \log(w+1)$ \cite[Example~5.4]{Odl95}. The proposition below brings down the number of comparisons required for verifying the unique minimizer condition to a ``mere'' $2^m-m-2$. 

For any non-empty subset $B=\{b_1,b_2,\ldots, b_{|B|}\}$ of $[m]$ with $|B| \ <m$, define $\cP_B\triangleq \bigl\{B^c,\{b_1\},\{b_2\},\ldots,\{b_{|B|}\}\bigr\}$ to be the partition of $[m]$ consisting of $|B|+1$ cells, of which $|B|$ cells are singletons comprising the elements of $B$. Note that if $|B|\ = m-1$, then $\cP_B = \cS$.
\begin{proposition}
For $m \ge 3$, let $\Omega = \{B \subset [m]: 1 \le |B| \, \le m-2\}$. The singleton partition $\cS$ is \\
\emph{(a)} a minimizer for $\mathbf{I}(X_{[m]})$ iff $\Delta(\cS) \le \Delta(\cP_B)$ $\forall\,B \in \Omega$; \\
\emph{(b)} the unique minimizer for $\mathbf{I}(X_{[m]})$ iff $\Delta(\cS) < \Delta(\cP_B)$ $\forall\,B \in \Omega$.
\label{prop:min}
\end{proposition}
\begin{IEEEproof}
We prove~(b); for~(a), we simply have to replace the `$>$' in \eqref{ifeqa} below with a `$\ge$'. 
 
The ``only if'' part is obvious. For the ``if'' part, suppose that $\Delta(\cS) < \Delta(\cP_B)$ for all $B \subset [m]$ with $1 \le |B| \ \le m-2$. Consider any partition $\cP$ of $[m]$, $\cP \ne \cS$, with $|\cP|\ \ge 2$. We wish to show that $\Delta(\cP) > \Delta(\cS)$. 

The following identity can be obtained from the definition in \eqref{def:DP} by some re-grouping of terms:
$$
\sum_{A \in \cP} |A^c| \, \Delta(\cP_{A^c}) = (|\cP|-1)[\Delta(\cP) + (m-1)\Delta(\cS)].
$$
Thus, we have
\begin{align}
\Delta(\cP) &= \frac{1}{|\cP|-1} \sum_{A \in \cP} |A^c| \, \Delta(\cP_{A^c}) - (m-1)\Delta(\cS) \notag \\
& > \frac{1}{|\cP|-1} \sum_{A \in \cP} |A^c| \, \Delta(\cS) - (m-1)\Delta(\cS) \label{ifeqa} \\
& = m\Delta(\cS) - (m-1) \Delta(\cS) \ = \ \Delta(\cS). \label{ifeqb} 
\end{align}
The inequality in \eqref{ifeqa} is due to the fact that at least one $A \in \cP$ is not a singleton cell, so that $\cP_{A^c} \ne \cS$, and hence, $\Delta(\cP_{A^c}) > \Delta(\cS)$ by assumption. To verify the first equality in \eqref{ifeqb}, observe that $\sum_{A\in\cP}|A^c| = \sum_{A\in\cP}\sum_{i \notin A} 1 = \sum_{i=1}^m \sum_{A \in \cP: i \notin A} 1 = m(|\cP|-1)$.
\end{IEEEproof}

\medskip

Next, we apply the above proposition to some interesting special cases.

Random variables $X_1,X_2,\ldots,X_m$, $m \ge 2$, are called \emph{isentropic} if $H(X_A) = H(X_B)$ for any pair of non-empty subsets $A, B \subseteq [m]$ having the same cardinality. Equivalently, $X_1,\ldots,X_m$ are isentropic if, for all non-empty $A \subseteq [m]$, the entropy $H(X_A)$ depends only on $|A|$. One obvious consequence of this definition is that for disjoint non-empty subsets $A, B \subset [m]$, the conditional entropy $H(X_A|X_B)$ only depends on $|A|$ and $|B|$. 

Clearly, i.i.d.\ rvs are isentropic. More generally, exchangeable rvs are isentropic --- rvs $X_1,X_2,\ldots,X_m$ are \emph{exchangeable} if for all permutations $\sigma$ of $[m]$, the joint distribution of $(X_1,X_2,\ldots,X_m)$ is the same as that of $(X_{\sigma(1)},X_{\sigma(2)},\ldots,X_{\sigma(m)})$. However, isentropic rvs need not be exchangeable. It may be verified that the rvs $X_1,X_2,\ldots,X_m$ in the PIN model defined on the complete graph $K_m$ (as defined in Section~\ref{sec:omni}) are not exchangeable when $m \ge 3$, but they are isentropic.

\setcounter{corollary}{0}
\begin{corollary}
If $X_1,X_2,\ldots,X_m$ , $m \ge 3$, are isentropic rvs, then $\cS$ is a minimizer for $\mathbf{I}(X_{[m]})$. 
\label{cor:isent}
\end{corollary}
\begin{IEEEproof}
For a partition $\cP$ of $[m]$ with $|\cP| \, \ge 2$, let us define
$$
\delta(\cP) \ \triangleq \ \frac{1}{|\cP|-1}\sum_{A\in \cP} H(X_{A^c}|X_A) \ = \ H(X_{[m]})-\Delta(\cP). 
$$
By virtue of Proposition~\ref{prop:min}(a), we need to show that $\delta(\cP_B) \le \delta(\cS)$ for all $B \in \Omega$.

For isentropic rvs, the quantity $H(X_B|X_{B^c})$, for any $B \subseteq [m]$, depends only on $|B|$. Hence, defining $g(k) \triangleq H(X_{[k]} | X_{[m]\setminus [k]})$ for $1 \le k \le m$, we can write $\delta(\cP_B) =  \frac{1}{|B|} g(|B|) + g(m-1)$ and $\delta(\cS) = \frac{m}{m-1}g(m-1)$.
Thus, we have to show that $\frac{g(|B|)}{|B|} \le \frac{g(m-1)}{m-1}$ for all $B \in \Omega$. This follows from the fact that for isentropic rvs, the function $g(k)/k$ is non-decreasing in $k$ --- see Appendix~A.
\end{IEEEproof}

Our second application of Proposition~\ref{prop:min} is to the PIN model. Recall from Section~\ref{sec:omni} that this model is defined on an underlying graph $\cG = ([m],\cE)$. From the way that the rvs $X_i$, $i \in [m]$, are defined, it is not difficult to verify that for any partition $\cP$ of $[m]$ with $|\cP|\, \ge 2$, we have
$$
\Delta(\cP) = \frac{|\cE(\cP)|}{|\cP|-1}, 
$$
where $|\cE(\cP)|$ denotes the number of edges $e=\{i,j\}\in \cE$ such that $i$ and $j$ are in different cells of $\cP$. This, in conjunction with Proposition~\ref{prop:min}, gives us a relatively simple criterion for verifying whether $\cS$ is a (unique) minimizer for $\mathbf{I}(X_{[m]})$. As an illustration, we apply this to the complete graph PIN model. 

\begin{corollary}
For the PIN model on the complete graph $K_m$, $m \ge 3$, the singleton partition $\cS$ is the unique minimizer for $\mathbf{I}(X_{[m]})$.
\label{cor:pin}
\end{corollary}
\begin{proof}
It is easy to see that for any non-empty $B \subsetneq [m]$, 
$|\cE(\cP_B)| = \binom{m}{2} - \binom{|B^c|}{2} = \frac12 |B|(2m-|B|-1)$. Hence,
$$
\Delta(\cP_B) = \frac{|\cE(\cP_B)|}{|B|} = \frac{2m-{|B|}-1}{2} \ge \frac{m}{2} = \Delta(\cS),
$$
with equality iff ${|B|}=m-1$, i.e., $\cP_B = \cS$. The result now follows from Proposition~\ref{prop:min}(b).
\end{proof}

\section*{Appendix A}
Here, we prove that for isentropic rvs $X_1,\ldots,X_m$, the function $\frac{1}{k} H(X_{[k]} | X_{[m]\setminus [k]})$, defined for $1 \le k \le m$, is non-decreasing in $k$. Define $g(k) = H(X_{[k]} | X_{[m]\setminus [k]})$. We show that the difference $kg(k+1) - (k+1)g(k)$ is always non-negative, from which the result follows. 

We have $g(k+1) = H(X_{[m]}) - H(X_{\{k+2,\ldots,m\}})$ and $g(k) = H(X_{[m]}) - H(X_{\{k+1,\ldots,m\}}) = g(k+1) - H(X_{k+1}|X_{\{k+2,\ldots,m\}})$. Thus, 
\begin{align*}
kg(k+1) & - (k+1)g(k) \\
&= (k+1) \, H(X_{k+1}|X_{\{k+2,\ldots,m\}}) - g(k+1).
\end{align*}
We need to show that the RHS of the equality is non-negative. This is straightforward:
\begin{align*}
g(k+1) &= H(X_{[k+1]}|X_{\{k+2,\ldots,m\}}) \\
& \le \sum_{i=1}^{k+1} H(X_{i}|X_{\{k+2,\ldots,m\}}) \\
& = (k+1) H(X_{k+1}|X_{\{k+2,\ldots,m\}}),
\end{align*}
since, for $1 \le i \le k+1$, $H(X_{i}|X_{\{k+2,\ldots,m\}}) = H(X_{k+1}|X_{\{k+2,\ldots,m\}})$ by isentropy.

\end{document}